\documentclass{article}
\usepackage[utf8]{inputenc}

\usepackage{graphicx}
\usepackage{amsmath, amsthm}
\usepackage{amsfonts}
\usepackage{algorithm, algorithmic}
\usepackage{authblk}

\newtheorem{definition}{Definition}

\newtheorem{theorem}{Theorem}
\newtheorem{proposition}{Proposition}

\newtheorem{example}{Example}

\def\Ag{A}

\def\PR{\mathit{PR}}  
\def\TR{\mathit{TR}}  

\title{Analysing Multi-Agent Systems using 1-safe Petri Nets}

\begin{document}

\author[1]{Federica	Adobbati}

\affil[1]{DISCo Università degli Studi di Milano-Bicocca}

\author[2,3]{{\L}ukasz Mikulski}

\affil[2]{Faculty of Mathematics and Computer Science, Nicolaus Copernicus University}
\affil[3]{Institute of Computer Science, Polish Academy of Sciences}

\date{}
\setcounter{Maxaffil}{0}
\renewcommand\Affilfont{\itshape\small}

\maketitle

\begin{abstract}
In the modelling and analysis of large, real systems, the main problem in their efficient processing is the size of the global model.
One of the popular approaches that address this issue is the decomposition of such global model into much smaller submodels and interaction between them.
In this paper we discuss the translation of multi-agent systems with the common-action-based synchronization to 1-safe Petri nets.
We prove that the composition in terms of transition systems is equivalent to the transition-based fusion of nets modelling different agents. 
We also address the issue of permanent disabling of some parts of the system by constraints implied by the synchronization and discuss the methods of solving it without the computation of the entire global model.
\end{abstract}

\section{Introduction}

A multi-agent system (MAS)
is composed of multiple decision-making agents interacting inside some environment.
They are highly distributed and concurrent systems.
Each agent performs actions aimed at reaching one's goal,
making decisions according to own strategy.
Thus, when the goal is common or beneficial for many agents,
they cooperate to achieve it.
On the other hand, they compete when their goals are in conflict.
Multi-agent systems originate from distributed artificial intelligence field \cite{ferber1999multi},
and are also studied by the formal methods and model-checking community.

The main challenge related to the analysis of real-live MAS is the complexity.
On one hand, the global model of such system is usually very large.
A simple voting model with four voters and one coercer discussed in \cite{STV+AGR} consists
of almost 300 000 states and 7 000 000 transitions.
On the other hand, the price for expressiveness, which allows to express properties
like Stackelberg or Nash equilibria,
is high computational complexity of model checking:
from 2EXPTIME-complete in the case of ATL*, to non-elementary for SL \cite{belardinelli2019strategy}.
It is worth  to note that ATL, being less expressive, allows to
model check strategic properties in polynomial time.

There are many methods that address the issue of both memory and time complexity for the analysis of MAS.
The methods reducing state space range from partial order reductions~\cite{JPSDM20}, where the analysed model is much smaller, 
through on-the-fly techniques~\cite{raimondi2005model}, hoping that the small part of the model is sufficient to verify the considered formula, ending with decomposition methods and assume-guarantee
techniques, where the local specification is checked on a small part of the system 
immersed in the abstraction of the global model~\cite{lomuscio2013assume}.
In order to speed up the computation, one can try to reason about the base of
approximations related with less expressive but more desirable models~\cite{JKKM19}, 
or reduce the space of checked strategies by their internal structure~\cite{Domino19}.
The research reported in this paper fits in the decomposition methods approach.

\paragraph{Related work}
Synthesis and  analysis of multi-agent systems is a well known topic in the literature, 
also in the context of Petri nets.
Pujari and Mukhopadhyay  in \cite{pujari2012petri} discuss
MAS as a discrete-event dynamic system (DEDS), and use Petri nets as a modeling tool to assess the structural properties of the system.  
Similarily, following DEDS concept, Lukomski and Wilkosz \cite{lukomski2010modeling}
show rules of modeling and analyzing the considered multi-agent
system with use of Petri nets.
Everdij et.al. \cite{everdij2006compositional} propose to use compositional specification power
of Petri nets in application to a multi-agent system.
Highly distributed air traffic operations system is considered
 and modelled using Stochastically and Dynamically Coloured Petri Nets.
The approach of Galan and Baker \cite{jafmas99} focuses on 
specifying and analyzing the conversations in a multi-agent
system. Conversations are specified using an automata
model and converted into a Petri net representation. Using
a Petri net analyzer, the conversations are checked for
consistency and coherency by testing liveness and safety of the resulting net.
Hiraishi in \cite{hiraishi2001petri} proposes $PN^2$ model, 
an extension of P/T nets, for the design and the analysis of multi-agent systems.

Another branch of the literature study shows a number of approaches using Petri nets to coordinate, organize or plan MAS behaviors.
In \cite{ziparo2011petri} a representation and execution framework for high level multirobot plan design, called Petri Net Plans (PNP), is proposed. 
PNP is based on Petri nets with a domain specific interpretation.
Places and transitions are partitioned into several classes of different interpretations.
A special case is a Petri net that has 
at most one token per place and edges of weight one.
As a central feature, PNPs allow for a formal analysis
of plans based on standard Petri net tools.
Scheduling by hierarchical structuring of the tasks performed by agents
is one of key ideas  in \cite{MOLINERO2011}.
They propose a multi-agent system that allows the user to
define a hierarchical structuring of the tasks that these agents perform, to plan a schedule involving parallel and sequential calling of the agents.
The agents are atomic or complex. 
Atomic agents are simple Petri nets performing a task, while
complex agents  are used to gather atomic (and/or other
complex agents) to conglomerate their individual
behaviour, and arrange their working order.
The authors of \cite{kouah2013synchronized} propose a framework for specifying multi-agent systems based on
Synchronized Petri Nets. 
It is an extension of Recursive Petri nets, 
 facilitating multi-agent system specifications
by concepts like: typed places,
transitions and tokens, synchronization points,
synchronization conditions, synchronization relations and
binding functions.

Finally, an approach that seems to be very close to ours uses Nested Unit Petri Nets (NUPN)\cite{garavel2019nested}.
One can see multi-agent systems as safe NUPNs of height 1 (taking every agent as a leaf unit and the whole system as
the root unit). The nets used in the examples and constructed with the use of a naive solution (taking a single place for every state of considered transition system) are even unit-safe. The considered problem of transitions disabled by synchronizations corresponds to the usual weak-liveness check for such nets. 

\paragraph{Contribution}
The main contribution of this paper is the proof of the operational equivalence between multi-agent systems and labelled 1-safe Petri nets. First we discuss how one can synthesize a Petri net of desired behaviour (see~\cite{BBD15} for comprehensive description) and utilize it to prepare a model for each agent separately. Following the idea of splitting transitions in the case of behaviour which cannot be covered by a 1-safe Petri net utilized in Petrify~\cite{cortadella1997petrify}, we first provide a naive solution. Then we prove that the composition of the set of agents defined in~\cite{JPSDM20} is equivalent to classical Petri net transition fusion~\cite{gomes2005structuring}.

As already noticed in~\cite{jamroga2021strategic}, asynchronous multi-agent systems with the composition based on the synchronization on actions are prone to all sorts of unwanted side effects. Some actions designed in the model of composed agents may be permanently disabled in the global model of the system. Such an artifact can surely be seen as an unwanted side effect of building the multi-agent system by the composition. As our second contribution, we address this issue and provide a procedure to check whether a particular action is permanently disabled without computing the entire global model of a considered multi-agent system. Although the proposed procedure does not reduce the complexity of the considered problem, it allows to perform calculations on a fragment of the considered model (which might be much smaller then the global model).

\section{The model}

\subsection{Asynchronous multi-agent systems}
In this section we recall an \emph{asynchronous multi-agent system} defined in~\cite{JPSDM20}.

\begin{definition}
An asynchronous multi-agent system (AMAS) consists of n agents $\Ag =\{1,\dots,n \}$.
Each agent is associated with a tuple
 
$A_i=(L_i,\iota_i,Evt_i,\PR_i,\TR_i,\mathcal{PV}_i,V_i)$, where
\begin{itemize}
    \item $L_i=\{l_i^1,\ldots,l_i^{n_i}\}$ is a set of local states;
    \item $\iota_i\in L_i$ is an initial state;
    \item $Evt_i=\{\alpha_i^1,\ldots,\alpha_i^{m_i}\}$ is a set of events in which agent $A_i$ can choose to participate;
    \item $\PR_i:L_i\to 2^{Evt_i}$ is a local protocol, which assigns 
    events to states in which they are available;
    \item $\TR_i:L_i\times Evt_i\to L_i$ is a local transition function, such that $\TR_i(l_i,\alpha)$ is 
    defined whenever $\alpha\in \PR_i(l_i)$;
    \item $\mathcal{PV}_i$ is a set of local propositions;
    \item $V_i:L_i\to 2^{\mathcal{PV}_i}$ is the valuation of local propositions in local states.
\end{itemize}
\end{definition} 

Since the model checking of AMAS is not in the scope of this paper, we are interested only in
transition systems which define the behaviour of AMAS, namely the tuples $(L_i,Evt_i,\TR_i,\iota_i)$.
Note, however, that local events of different agents may not be disjoint. 
Events which are present in more than one event set $Evt$ require participation of more than one 
agent, namely those agents synchronize on such events. 

\begin{example}
Fig. \ref{fig:tramas} represents an AMAS with three agents. The events $n1, m1,$ $n2, m2$ are shared
by two agents,
and 
require the participation of 
both
of them to occur,
whereas the events $n3$ and $m3$ are local, and depends on a single agent. 
\end{example}
\begin{figure}
    \centering
    \includegraphics[width=0.5\textwidth]{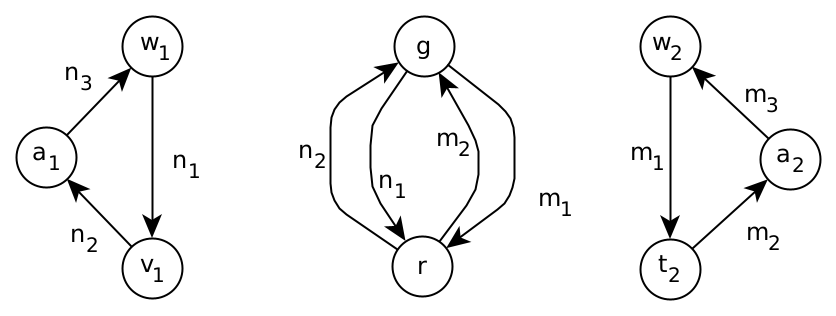}
    \caption{Example of three local models for agents related to Train-Gate-Controller benchmark from \cite{JPSDM20}.}
    \label{fig:tramas}
\end{figure}

\subsection{1-safe Petri nets}
Petri nets were introduced by Carl Adam Petri in his PhD thesis \cite{Petri62} as a 
formal graphical model to represent and analyse concurrent systems. 
In this section we provide some basic definitions that will be useful in the rest of the 
paper, for an extensive overview about Petri nets and their applications we refer to \cite{murata89,peterson77}. 

A \emph{plain} net is characterized by a set of \emph{places} or \emph{conditions} $P$, represented as circles, by a set of \emph{transitions} $T$, represented as squares, 
and by a  \emph{flow relation} between them $F\subseteq (P\times T) \cup (T \times P)$, represented with arcs\footnote{In more general definition one can define a weight function $W:(P\times T) \cup (T \times P)\to\mathbb{N}$ instead of flow relation $F$
.}. 

For each element $x\in P\cup T$, its preset is $^\bullet x = \{y\in P\cup T : (y, x)\in F\}$, and its postset is $x^\bullet = \{y\in P\cup T : (x, y)\in F\}$. 
For each transition $t\in T$, we assume that its preset and its postset are non-empty, 
i.e. $^\bullet t \neq \emptyset$ and $t^\bullet \neq \emptyset$. 
The elements in $^\bullet t$ are also called \emph{preconditions} of $t$, and the elements in $t^\bullet$ 
are also called \emph{postconditions}.

A net system is a quadruple $\Sigma = (P, T, F, m_0)$, where $P, T, F$ are the elements of the net, 
and $m_0: P \rightarrow \mathbb{N}$ is the \emph{initial marking}. 
A transition $t\in T$ is enabled in a marking $m$ if for each $p\in {^\bullet t}$, $m(p) \geq 0$. 
If a transition is enabled, it can \emph{occur} or \emph{fire}, and its occurrence generates a 
new marking $m'$ defined as follows. 
\[
    m'(p) =
  \begin{cases}
    m(p)-1 & \text{for all \(p \in {^\bullet t} \setminus t^\bullet\)}\\
    m(p)+1 & \text{for all \(p \in t^\bullet \setminus {^\bullet t}\)}\\
    m(p)   & \text{in all other cases.}
  \end{cases}
\]
In symbols, $m[t\rangle$ denotes that $t$ is enabled in $m$, while $m[t\rangle m'$ denotes 
that $m'$ is the marking produced from the occurrence of $t$ in $m$. 
A marking $m$ is reachable in a net system ${\Sigma = (P, T, F, m_0)}$ if 
there is a sequence of transitions (called a \emph{firing sequence}) 
$t_1...t_n$ such that $m_0[t_1\rangle m_1...m_{n-1}[t_n\rangle m$. 
The set of all the reachable markings is denoted with $[m_0\rangle$. 
A transition $t\in T$ is 1-live if there is a marking $m\in [m_0\rangle$ such that $m[t\rangle$. 

Let $m$ be a reachable marking, and $t_1, t_2\in T$ be two transitions enabled in $m$: 
$t_1$ and $t_2$ are in \emph{conflict} in $m$ if $^\bullet t_1 \cap {^\bullet t_2} \neq \emptyset$; 
$t_1$ and $t_2$ are concurrent in $m$ if $^\bullet t_1 \cap {^\bullet t_2} = \emptyset$ and $t_1^\bullet \cap t_2 ^\bullet = \emptyset$. 

In this paper we work with the class of 1-safe net systems. 
A net system is \emph{1-safe} if, for each $m\in [m_0\rangle$ and for each $p\in P$, $m(p) \leq 1$. 
In a 1-safe system, each marking can (and will) be considered as a set of places, and each place can 
be interpreted as a proposition, that is true if the place belongs to the marking, and false otherwise. 

A \emph{labelled} Petri net $\Sigma_\lambda = (\Sigma, \lambda)$ is a 1-safe net system with a function 
$\lambda: T \rightarrow \Lambda$, where $\Lambda$ is a set of labels. Abusing the notation, for each 
$T' \subseteq T$ subset of $T$, we will denote with $\lambda(T') = \{\alpha\in \Lambda : \exists t\in T': \lambda(t) = \alpha\}$ 
the set of labels of the elements in $T'$. The set $\lambda(T)$ is the \emph{alphabet} of $\Sigma_\lambda$. 

The sequential behaviour of a labelled Petri net 
can be described by an initialized labelled transition system, 
where each state corresponds to a reachable marking, 
and each arc is labelled by the label of the transition leading from the source marking to the target one.

\begin{definition} 
	Let $\Sigma = (P, T, F, m_0, \lambda)$ be a  labelled Petri net, 
	its \emph{marking graph} is a quadruple $MG(\Sigma) = ([m_0\rangle, \lambda(T), Ar, m_0)$
	where  
	
	$Ar= \{ (m, \lambda(t), m')\ | \ m, m' \in [m_0\rangle, \  t \in T,\,  m[t \rangle m' \}$. 
\end{definition}

\subsection{Synthesis of 1-safe Petri nets from AMAS}
For each agent in an AMAS, we can always obtain a 1-safe labelled Petri net: 
the agent in the AMAS can be considered as the marking graph of the net, 
therefore the net can be found through a synthesis procedure. 
The classical techniques for the synthesis of 1-safe net systems are based on 
the research of \emph{regions} \cite{BBD15}. 
Let $A_i = (L_i, Evt_i, \TR_i, \iota_i)$ be a labelled transition system. 
A \emph{region} is a subset of states $r\subseteq L_i$ such that, for each 
$e\in Evt_i$, 
one of the following conditions holds: 
(1) $e$ enters the region, i.e. for each arc labelled with $e$ from $s_1$ to $s_2$, 
$s_1 \not \in r \wedge s_2 \in r$; (2) $e$ leaves the region, i.e. for each arc labelled 
with $e$ from $s_1$ to $s_2$, $s_1 \in r \wedge s_2 \not\in r$; (3) $e$ does not cross the 
border of the region, i.e. for each arc labelled with $e$ from $s_1$ to $s_2$, 
$s_1 \in r \wedge s_2 \in r$, or $s_1 \not\in r \wedge s_2 \not\in r$. 
In order to synthesize a 1-safe net system, every label in $Evt_i$ is translated into a transition of the net system, 
and every region into a place. The flow relation is determined as follows: for each 
transition $e$, for each place $r$, $r$ is a precondition of $e$ if $e$ leaves $r$ in 
$A_i$, $r$ is a post-condition of $e$ if $e$ enters $r$ in $A_i$. 
If $e$ does not cross the border of $r$, and for each arc labelled with $e$ from 
$s_1$ to $s_2$, we have $s_1, s_2\in r$, then we can see $r$ as both a precondition and a 
postcondition of $e$, and add a self loop to the net. 
Otherwise, there is no flow relation between $r$ and $e$ in the net. 

Not every transition system can be synthesized into a 1-safe net. 
In particular, we can synthesize a 1-safe net system from a transition system if, 
and only if, the set of regions of the transition system satisfies the so 
called \emph{state separation property} (SSP) and \emph{event-state separation property} (ESSP): 
\begin{align*}
    \forall s, s'\in L_i \quad s\neq s'  \quad \rightarrow \exists r\in R : (s\in r \wedge s' \not\in r) \lor (s\not\in r \wedge s\in r)\quad (SSP)\\
    \forall e \in E, \, \forall s\in S : e \text{ is not outgoing from } s, \rightarrow \exists r : s\not\in r 
    \wedge e \text{ leaves } r \quad(ESSP)
\end{align*}
\begin{example}
Consider the central transition system on the upper part of Fig.~\ref{fig:glnet}.  
Since there is an outgoing transition labelled as $3$ from all the states, the only 
regions in this system are trivial: the empty set, and the set of all the states. 
These regions do not allow to separate any pair of states, therefore SSP is not satisfied, and the transition system is not synthesizable with a 1-safe net without labels. 
\end{example}
However, it is always possible to obtain a labelled 1-safe system, by allowing the net to  
have more transitions with the same label. 
In this case, we can split the transitions of $A_i$ with the same label into subgroups, and 
look for regions as if each group had a different label. 
This generates a set of different transitions in the net sharing the same label. 
To obtain such a net is always possible, since we could consider subsets formed by 
single arcs: if each arc of the transition system is considered as if it had different 
labels from the others, it is easy to verify that the set of regions satisfies the 
separation properties SSP and ESSP. 
Furthermore, each state of $A_i$ is a region and therefore 
can be translated into a place of the synthesized net. 
As showed in \cite{BBD15}, the minimal 
regions with respect to inclusion are sufficient for the synthesis, therefore  we can consider the states of $A_i$ as all and only the places of the net. 
In what follows we will transform each agent of the AMAS in this way. 

For each $i$,  
we can represent each agent as a labelled Petri net. 
In particular, for each agent 
$A_i = (L_i,Evt_i,\TR_i,\iota_i)$, 
the associated Petri net is 
defined as $\Sigma_i = (L_i, T_i, F_i, \iota_i, \lambda)$, where:
\begin{itemize}
    \item $L_i$ is the set of places that coincides with the set of local states in $A_i$; 
        \item 
$T_i$ is the set of transitions, and there is one for each element in $\TR_i$; 
    \item 
$F_i$ is the flow relation, fully determined by $\TR_i$; 
\item $\iota_i$ is the initial 
marking, that coincides with the initial state of $A_i$ ; 
    \item  
$\lambda: T_i \rightarrow Evt_i$ is the  labelling function, associating 
every transition of the net with the label of the corresponding arc on the agent. 
\end{itemize}
\begin{example} 
Consider the AMAS in Figure~\ref{fig:tramas}, described also in \cite{JPSDM20}.
In Figure~\ref{fig:tr2pn} the agents of the AMAS are represented as Petri nets. 
\begin{figure}
    \centering
    \includegraphics[width=0.7\textwidth]{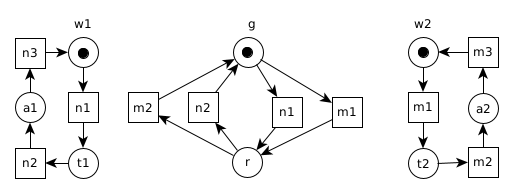}
    \caption{1-safe Petri nets for transition systems depicted in Fig.~\ref{fig:tramas}.}
    \label{fig:tr2pn}
\end{figure}

\end{example}

\section{Composition}

\subsection{Composition of AMAS}
In this section we recall from~\cite{JPSDM20} the definition of \emph{canonical interleaved interpreted system},
which is a composition of agents being parts of an asynchronous multi-agent system 
with synchronizations on common events.
Note that since we are not interested in model checking, we would concentrate only on the behaviour
of the multi-agent system (putting apart the propositional variables). 
\begin{definition}
A canonical interleaved interpreted system (canonical IIS) is an AMAS extended with a tuple $(St,Evt,\TR,\iota)$ where:
\begin{itemize}
    \item $St\subseteq L_1\times\ldots\times L_n$ is a set of global states;
    \item $Evt = \bigcup_{i\in\{1,\ldots,n\}} Evt_i$ is a set of events;
    \item $\TR:St\times Evt\to St$ is a (partial) global transition function, where
    $\TR((l_1,\ldots,l_n),\alpha)=(l'_1,\ldots,l'_n)$ if \ 
    $\TR_i(l_i,\alpha) = l'_i$ for all $i$ where $\alpha\in Evt_i$ and
    $\TR_i(l_i,\alpha) = l_i$ otherwise;
    \item $\iota=(\iota_1,\ldots,\iota_n)$ is an initial state.
\end{itemize} 
\end{definition}
Given a canonical IIS $I$, some of its states may not be reachable through any execution, 
due to the restrictions given by the synchronizations, 
and therefore also the transitions outgoing from these states can never be executed. 
We will denote with $I_r$ the canonical system where these unreachable states and transitions have been pruned. 

By definition, the number of states in the IIS grows exponentially with the number of agents, 
therefore limiting the number of compositions when studying the properties of the system may help in the analysis. 

\subsection{Composition of 1-safe Petri nets}
%
Let $\Sigma_1 = (L_1, T_1, F_1, \iota_1, \lambda), ..., \Sigma_n = (L_n, T_n, F_n, \iota_n, \lambda)$ be the set of Petri net agents. 
We then construct a \emph{global} net $\Sigma = (P, T, F, m_0, \lambda)$, showing the 
interaction of the agents. 
The set of places $P$ of $\Sigma$ is the union of the sets of places $L_i$. 
For each label $\alpha \in \lambda(T_i)$, for each agent $\Sigma_i$, 
let $T_i^\alpha = \{t\in T_i : \lambda(t) = \alpha\}$ be the set of transitions 
labelled with $\alpha$. The set of transition of $\Sigma$ is defined as 
$T = \bigcup_{\alpha\in \lambda(T_i)}\bigotimes_{i \in \{1,...,n\}}T_i^\alpha$. 
The flow relation is determined in this way: for each transition 
$t \in T$, and each place $p\in P$ there is an arc from $p$ to $t$ iff 
there is a $\Sigma_i$ and $t_j\in T_i$ such that $p\in L_i$, $t_j$ is a component 
in $t$, and $(p, t_j) \in F_i$; analogously for the arcs from $t\in T$ to 
$p\in P$. 
The initial marking $m_0$ is the union of all the elements $\iota_i$, with 
$i\in \{1, ..., n\}$. 
The labelling function $\lambda$ associates every transition $t\in T$ to the 
label of all its component, that is unique by construction. 
We denote the alphabet of $\Sigma$ with $\lambda(T)$.

By construction, each place in $\Sigma$ belongs at most to one agent, 
whereas the transitions can be shared. 
Note that some of the transitions may be enabled in no reachable marking, 
and therefore are not 1-live. 
As we discussed in the previous section, the same problem happens when we consider the composition of AMAS, because some 
states may not be reachable from the initial state $\iota$, due to the synchronization 
constraints. 
The problem of finding these transitions is discussed in detail in Sec.~\ref{sec:1-live}. 
\begin{figure}
    \centering
    \includegraphics[width=0.5\textwidth]{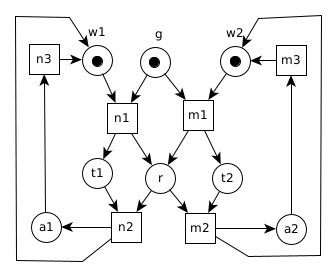}
    \caption{Global Petri net model resulting from the composition of the nets depicted in Fig.~\ref{fig:tr2pn}.}
    \label{fig:tr2glpn}
\end{figure}
\begin{example}
\label{ex:small_parts}
Fig.~\ref{fig:tr2glpn} represents the composition 
of the three agents from Fig. \ref{fig:tr2pn}. 
In this model, for each agent, each label appears 
only in one transitions, therefore in the global 
net in Fig. \ref{fig:tr2glpn} each transition has 
a different label. 

This is in general not the case, as we can see 
in Fig.\ref{fig:glnet}. In the latter case both the 
Petri net agents and the global model have the 
same label shared between more transitions. 
\begin{figure}
    \centering
    \includegraphics[width = 0.85\textwidth]{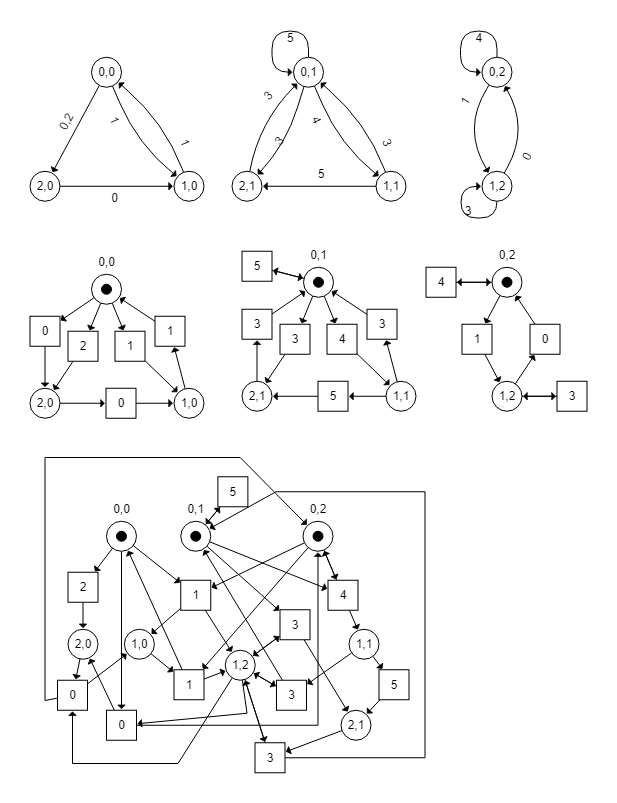}
    \caption{Three agents described as automata, the same agents represented with Petri nets, and the global Petri net.}
    \label{fig:glnet}
\end{figure}
\end{example}
Let $I_r$ be the canonical IIS where all the unreachable states and transitions have been removed. 
The following proposition shows that synthesis and composition are commutative, 
i.e synthesizing Petri net agents from the AMAS and then composing them is equivalent to 
construct the composition of AMAS and then synthesizing a Petri net. 

\begin{proposition}
Let $S$ be an AMAS, $\Sigma = (P, T, F, m_0, \lambda)$ be global Petri net 
constructed as described above, and $I$ the canonical IIS of $S$. 
The transition system of $\Sigma$ is isomorphic to $I_r$.
\end{proposition}
\begin{proof}
The initial marking of $\Sigma$ is the initial state of $I$ by the construction. 
Let $m$ be any reachable marking in $\Sigma$. By the construction, $m$ is 
a set of $n$ places, each of them from  a different set $L_i$, with $i\in\{1, ..., n\}$. For each $i\in\{1, ..., n\}$ we denote as $l_i$ the element of $m$ in the set $L_i$.
Let $t\in T$ be a transition such that $\lambda(t) = \alpha$. 
By construction, $t$ is enabled in $m$ iff, for each agent $\Sigma_i$, 
$i\in\{1, ..., n\}$ such that $\alpha \in \lambda(T_i)$, there is a 
transition $t_i\in T_i$ such that $\lambda(t_i) = \alpha$ and $t_i$ enabled 
in $l_i$.   
This condition is equivalent to the one of the global transition function 
defined for $I_r$, therefore for each marking, the set of outgoing transitions and states reachable in one step is the same. Hence the two models are 
isomorphic. 
\end{proof}
Note that, for simplicity, in the above proposition we consider the specific 1-safe labelled systems. 
One can easily repeat similar reasoning for any 1-safe labelled systems which are synthesized
from transition systems describing behaviours of particular agents and their compositions.

\section{1-liveness of transitions}
\label{sec:1-live}
In this section we discuss how to find transitions that are not 1-live 
on the global net. 
This is known to be a PSPACE-complete problem \cite{JLL77,esparza96}. 
We propose an algorithm that, in some cases, does not need to construct the 
global net in order to verify whether a transition is 1-live, but uses a 
smaller subnet. If this is possible, some computation is saved, since the 
complexity of the problem depends on the dimension of the net. 
In the worse case, the algorithm reconstructs the global system, and 
checks 1-liveness on it. 

Consider the net in Fig.~\ref{fig:glnet}. 
All the transitions labelled with 0 will never be enabled. 
A simple way to find these transitions consists in computing the 
marking graph of the net and check its labels. 
The transitions that do not appear in the marking graph will never 
be enabled. 
However, this can be computationally very expensive, since having all 
the agents may increase the level of concurrency, and therefore the 
size of the transition system.

A first alternative idea could be to find some of the transitions that will 
never be enabled by composing for each label all the agents 
sharing it. For example, the labels 1 and 0 are shared by the first 
and third agent in Fig. \ref{fig:glnet}. Fig. \ref{fig:comp13} 
shows the composition of the two agents, and the marking graph of 
this reduced net. The part coloured in red shows the transitions 
that cannot be enabled.  
\begin{figure}
    \centering
    \includegraphics[width = 0.85\textwidth]{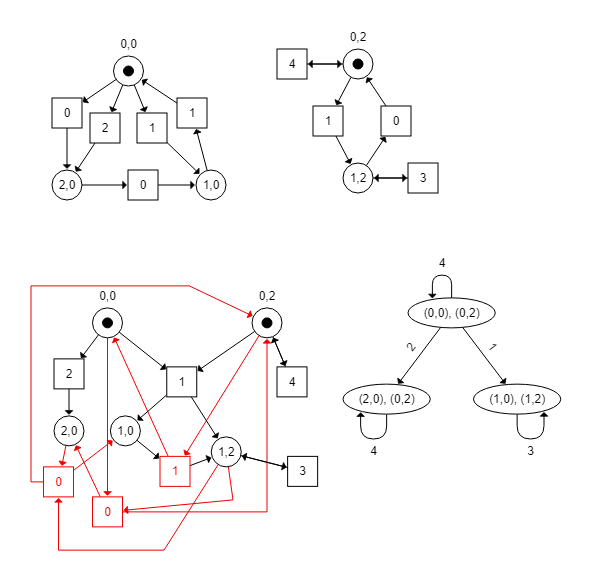}
    \caption{Composition of the first and third transition system depicted 
    in Fig.~\ref{fig:glnet}.}
    \label{fig:comp13}
\end{figure}
If a transition cannot occur in the net obtained composing only 
the agents sharing its label, then it cannot occur in the global net, 
since adding new components can only add the number of constraints, due 
to the synchronization requirements.

Unfortunately, this is only a necessary condition to identify 
transitions that cannot fire, but it is not sufficient. 
To see an example, consider the agents represented in Fig.~\ref{fig:compnot}. 
In the upper part of the figure, we see three agents sharing some of their 
labels. 
Below it, there is the global net, where the unreachable parts are coloured 
in red, and its transition system. At the bottom of the figure, we find 
the composition of the second and third agent, and its transition system. 
The label $d$ is shared only by the second and third agent, and by composing 
only them, it seems to be possible to fire it. 
However, this is not true, as we can see in the global transition system. 
This happens because in the reduced composition for label $d$, $b$ is 
considered as a label of the third agent only, whereas in the global system, 
it must synchronize with the transition of the first agent. 
We can see from the reduced transition system that $d$ must occur after 
$c$ and $b$, but since $b$ cannot fire in the global net, also $d$ cannot 
be reached. 

\begin{figure}
    \centering
    \includegraphics[width = 0.85\textwidth]{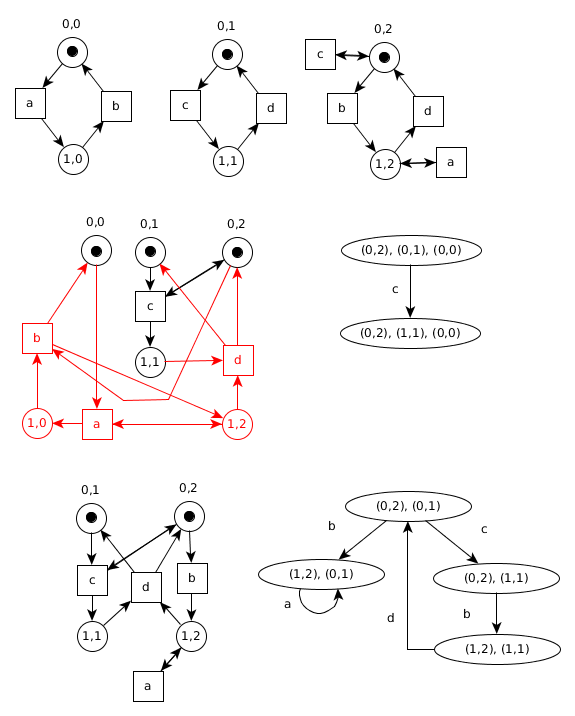}
    \caption{Example of a transition that from the reduced composition 
    falsely seems to be 1-live.}
    \label{fig:compnot}
\end{figure}

This suggests us another element to check whether a transition can be enabled 
without constructing the transition system of the global net, consisting in composing 
all the agents sharing a certain label, and all the labels in a minimal path 
leading to it from the initial state. 
\begin{definition}
Let $MG(\Sigma)$ be the marking graph of $\Sigma$. 
The firing sequence $t_1...t_n$ is \emph{minimal} if for each $i,j < n$, $m_i \neq m_j$, 
where $m_{i-1}[t_i\rangle m_i$.
\end{definition}
Let $\Lambda \subseteq \lambda(T)$ be any subset of labels in the global net $\Sigma$. 
We will denote with $\Sigma^\Lambda$ the net obtained by 
composing all the agents with at least an element of $\Lambda$ in their alphabet. 

Let $\alpha$ be any label on the system, and $T_\alpha =\{t_j\in T : \lambda(t_j) = \alpha\}$.  
Algorithm~\ref{alg:1-live} shows how to check that 
$t_j\in T_\alpha$ is 1-live in $\Sigma$, without computing the entire system. 
By applying the algorithm to each $t_j\in T_\alpha$, we can discover which transitions of $\Sigma$ are not 1-live, 
and therefore can be removed without changing the behaviour of the net. 

The algorithm takes as input $t_j$, the set of all the agents in the system, and a set of labels $\Lambda$, 
and returns true if there is 
a firing
sequence of transitions that enable $t_j$, false otherwise.  
In addition, if one exists, the algorithm returns the sequence $\pi$ of 
transitions leading to $t_j$.
In the first call $\Lambda = \{\alpha\}$. 

The algorithm has a recursive structure. 
The first step consists in computing the minimal paths from the initial state of $\Sigma^\Lambda$ 
to $t_j$ (this is done by the function comp\_min\_paths).
If there are no minimal paths in $MG(\Sigma^\Lambda)$, 
then it returns false, since $\Sigma^\Lambda$ does not need to be further explored. 
Otherwise, it selects a minimal path $\pi$, through the function select\_path. 
Let $\lambda(\pi)$ be the set of labels of the transitions in $\pi$. 
If $\Sigma^{\lambda(\pi)} = \Sigma^\Lambda$, 
then the algorithm returns true, since we found a path that can be executed on $\Sigma$ and enables $t_j$. 
When this happens, there is no need to look for alternative paths, 
and the computation can stop and return true. 
This is not the case if a recursive call returns false, since 
the unreachability of $t_j$ may be due to a 
wrong choice of the path in one of the previous steps. 
Then, we need to check if, in previous calls, other paths could have been chosen, 
leading to different subsystems, and check if $t_j$ is reachable in them. 
If $t_j$ is not reachable from any path, 
then we can conclude that $t_j$ is not 1-live. 
\begin{algorithm}
\caption{Check if $t_j$ is 1-live}
\begin{algorithmic} \label{alg:1-live}
\STATE\textbf{procedure} check\_1liveness($t_j, \{\Sigma_i: i\in \{1,...,n\}\}, \Lambda )\in \{$\textbf{true}, \textbf{false} $\} \times \Pi$
\STATE \quad $\Pi = $comp\_min\_paths$(MG(\Sigma^\Lambda), t_j)$
\STATE \quad \textbf{if} $\Pi == \emptyset$
\STATE \quad \quad \textbf{return false}, $\emptyset$
\STATE \quad \textbf{end if}
\STATE \quad \textbf{while} $\Pi \neq \emptyset$
\STATE \quad\quad $\pi = $select\_path$(\Pi)$
\STATE \quad\quad $\Pi = \Pi \setminus \{\pi\}$
\STATE \quad\quad \textbf{if} $\Sigma^\Lambda == \Sigma^{\lambda(\pi)}$
\STATE \quad\quad\quad $\pi' = \pi$
\STATE \quad \quad\quad \textbf{return true}, $\pi$
\STATE \quad\quad \textbf{end if}
\STATE \quad\quad $r, \pi' =$ check\_1liveness($t_j, \{\Sigma_i : i\in \{1, ..., n\}\}, \lambda(\pi)$)
\STATE \quad \quad \textbf{if} $r ==$ \textbf{true}
\STATE \quad\quad \quad \textbf{return true}, $\pi'$
\STATE \quad \quad \textbf{end if}
\STATE \quad \textbf{end while}
\STATE \quad \textbf{return false}, $\emptyset$
\STATE \textbf{end procedure}
\end{algorithmic}
\end{algorithm}

By construction, for each transition $t_i$ in the sequence $\pi$ returned by Algorithm~\ref{alg:1-live}, 
the set of preconditions and the set of postconditions are the same in $\Sigma$ and in $\Sigma^{\lambda(\pi)}$.
\begin{theorem}
Algorithm~\ref{alg:1-live} is correct, i.e. for each transition $t_j\in T$, the algorithm returns true iff 
$t_j$ is 1-live in $\Sigma$.
\end{theorem}
\begin{proof}
As the first step, we show that if the algorithm returns true, then $t_j$ is executable in $\Sigma$, 
and in particular the path $\pi = t_1...t_j$ returned by the algorithm is 
a firing sequence
of $\Sigma$. 
We proceed by induction, starting to show that $t_1$ is enabled in $m_0$. 
By contradiction, let us suppose that $t_1$ is not enabled in $m_0$. Then, there must be a precondition 
$p\in {^\bullet t_1}$, such that $p\not \in m_0$. By the construction, all the elements in $^\bullet t_1$ 
come from agents that have transitions labelled with $\lambda(t_1)$, and all these agents are included in 
$\Sigma^{\lambda(\pi)}$; hence, if $t_1$ is enabled in the initial state of $\Sigma^{\lambda(\pi)}$, it must be enabled 
also in $m_0$. 
Let $\pi_i = t_1...t_i$, $i<j$, be a prefix of the firing sequence $\pi$  
and $m_i$ the 
state reached in $\Sigma$ after executing $\pi_i$. We show that $t_{i+1}$ is enabled in $m_i$. 
By contradiction, let us suppose that $t_{i+1}$ is not enabled in $m_i$. Then there must be a place 
$p\not\in m_i$ and such that $p\in {^\bullet t_{i+1}}$. This place must be also in $\Sigma^{\lambda(\pi)}$, since, 
by the construction, $\Sigma^{\lambda(\pi)}$ includes all the agents with transitions labelled $\lambda(t_{i+1})$, 
and the preconditions of $t_{i+1}$ on $\Sigma$ cannot belong to any other agent.  
Let $m^{\lambda(\pi)}_0t_1m^{\lambda(\pi)}_1...t_im^{\lambda(\pi)}_i$ the sequence of states and transitions obtained 
by firing the sequence $t_1...t_i$ in $\Sigma^{\lambda(\pi)}$;
$p\in m^{\lambda(\pi)}_i$, since $t_{i+1}$ can fire after $\pi_i$ in $\Sigma^{\lambda(\pi)}$, and there must be 
an index $k\leq i$ such that $p\in m^\alpha_r$ for each $r\geq k$. If $k = 0$, then $p\in m_0$, since 
$m_0^{\lambda(\pi)} \subseteq m_0$, by the construction. If $k > 0$, then $p \in t_k^\bullet$, and $p\in m_k$. 
Since the set of preconditions and postconditions of $t_{k+1}...t_i$ is the same in $\Sigma$ and 
$\Sigma_{\lambda(\pi)}$, if $p\in m^{\lambda(\pi)}_i$ after the execution of $t_{k+1}...t_i$, then $p\in m_i$ after firing the same sequence. 

As the second step, we need to prove that if the algorithm returns false, then $t_j$ is not 1-live in 
$\Sigma$. This follows from the observation that adding agents to the system can only restrict 
the possibility of the transitions in $\Sigma^\alpha$ to occur by adding synchronizations 
constraints. Therefore, if $t_j$ is a transition in $\Sigma^\Lambda$, but there is no sequence 
in $\Sigma^\Lambda$ enabling $t_j$, a fortiori there cannot be any sequence in $\Sigma$.
\end{proof}

\begin{proposition}
Algorithm~\ref{alg:1-live} terminates after a finite number of steps.
\end{proposition}
\begin{proof}
The thesis follows from the finiteness of the number of agents in the system, and of the number of 
minimal paths.
\end{proof}

Algorithm~\ref{alg:1-live} does not guarantee that the system $\Sigma^{\lambda(\pi)}$ to check will be 
smaller than $\Sigma$, since the two systems may coincide (for example in the systems in Fig. \ref{fig:glnet} and Fig. \ref{fig:compnot}). 
However, in distributed systems in which each agent interacts with a small 
subset of other agents of the entire system, it may become a convenient technique.  
An example of such a system could be represented by a social network, where the number of users is huge, 
but each of them has a limited number of connections. 
A toy example is represented in Fig.~\ref{fig:conv}. In this case concurrency enlarges the size of 
the global transition system, while does not affect the reduced systems. 
\begin{figure}
    \centering
    \includegraphics[width = 0.8\textwidth]{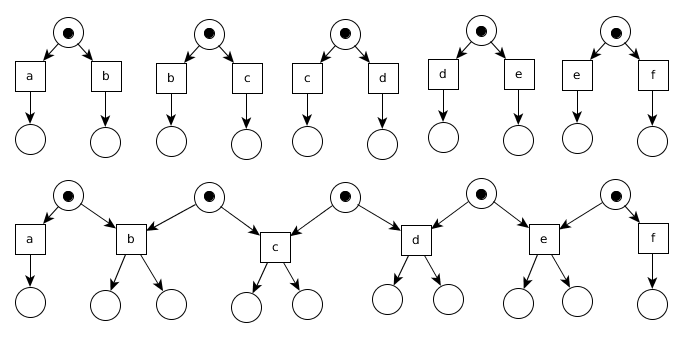}
    \caption{Example of the system where each agent interact with a small subset of other agents.}
    \label{fig:conv}
\end{figure}

\begin{figure}[ht]
    \centering
    \includegraphics[width = 0.8\textwidth]{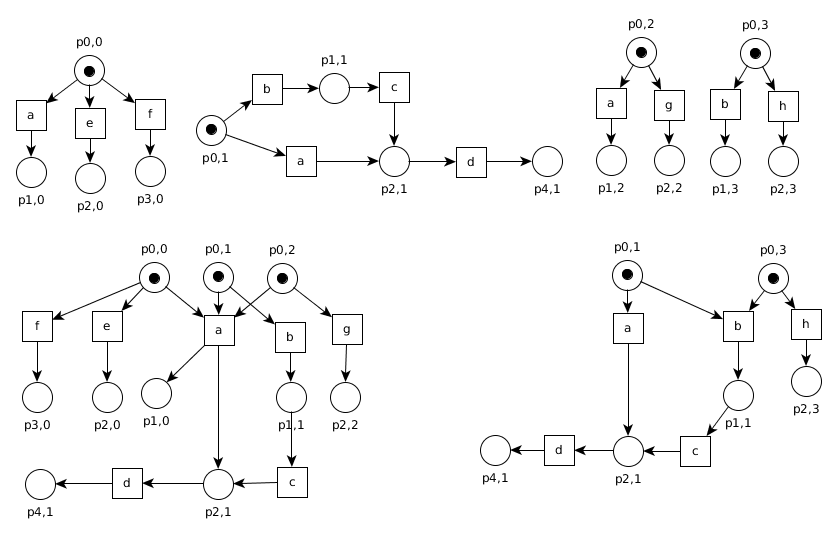}
    \caption{Net with two possible paths to reach $d$ and the alternatives compositions.}
    \label{fig:select-path}
\end{figure}
The required computation can also be reduced by choosing proper heuristics for the function 
select\_path, so that the more convenient paths are selected to be analysed first. 
A possible criterion could be to select first the paths requiring to add the minor number of new agents, 
even when they are longer and with more labels than others. 
Consider for example the system of agents in Fig.~\ref{fig:select-path}. 
The four agents are represented at the top of the figure, and two of their compositions at the bottom.
Transition $d$ belongs only to agent 1 (the second from the left), and there are two minimal 
sequences reaching it: $ad$, and $bcd$. 
Although $ad$ is shorter, the label $a$ is shared by two more agents: agent 0 and agent 2, whereas $c$ 
belongs only to agent 1, and $b$ is shared with agent $3$ only. 
The two compositions of agents represented in Fig.~\ref{fig:select-path} shows $\Sigma^{\{a, d\}}$ (on the left), 
and $\Sigma^{\{b,c,d\}}$ (on the right). It is easy to see that $\Sigma^{\{b,c,d\}}$ has less reachable states 
than $\Sigma^{\{a,d\}}$, and it is sufficient to decide the 1-liveness of the transition labelled with $d$.

Note that, for simplicity, once more we have considered only the specific 1-safe labelled systems.
If we restrict ourselves to such situation, most of the reasoning can be repeated without using
Petri nets. However, the real performance improvement is revealed when we utilize more sophisticated
synthesis algorithms based on region theory.
One can expect that in such cases the number of transitions with the same label 
appearing in a single local model would be smaller. 
Moreover, such approach gives a chance to decompose local models into sequential components 
(see~\cite{rozenberg1996elementary}) and use them instead of entire modules.

\section{Summary}
In this paper we provided a 1-safe Petri net framework to support the reasoning about multi-agent systems.
We have shown that it faithfully reflects the transition-based composition of the agents
and illustrates the usage in the case of checking 1-liveness of transitions.
Continuing this thread of research we plan to take a closer look to the ideas highlighted at the end of 
Section~\ref{sec:1-live}.
The mentioned semi-automatic decomposition of the agents' behaviour into smaller subsystems 
has a large potential not only in the case of checking 1-liveness, 
but also in preparing the environment (formal assumption subsystem) for the assume-guarantee reasoning.

Taking into account the similarity of labelled 1-safe Petri net models of AMAS described in paper with
Nested Unit Petri Nets, we would like to examine the effectiveness of proposed algorithm comparing with the 
liveness checks in the existing tools for analysing NUPNs like EVALUATOR \cite{mateescu2008specification}.

However, in the future we would like to consider systems that are k-safe. The reason is twofold,
there are synthesis methods which work fine without the restriction of 1-safeness. The resulting nets usually require less use of transition splitting, hence are smaller. Moreover, such approach is much more natural in planned by us analysis of asynchronous systems, where the synchronization is data-oriented (not action-oriented as in the approach presented in this paper).

Recently, an approach to automated synthesis of MAS 
based on satisfiability and model checking tools,
has been presented in \cite{msatl}.
After specifying the constraints and the strategic properties to be met,
the tool, exploiting monotonic theory for ATL \cite{KR20},  
looks for the model satisfying all requirements.
To reason about strategic abilities of MASs  
variants of  ATL*  and SL logic are considered \cite{belardinelli2019strategy,KR21}.
As a future work, we plan to utilize the described framework both in direct model checking using
1-safe Petri nets, and in the synthesis of MAS satisfying desired ATL formulas. 
 
\section*{Acknowledgement}
The first author was supported by the Italian MUR.
  The second author was supported by the National Centre for Research and Development, Poland (NCBR),
		and by the Luxembourg National Research Fund (FNR), under the PolLux/FNR-CORE project STV (POLLUX-VII/1/2019). 

\bibliographystyle{plain}
\bibliography{arxiv_version}

\end{document}